%% file: ms.tex
\documentclass[envcountsame]{llncs}
\usepackage[colorlinks]{hyperref}

\input{macros}

\usepackage{version}
\usepackage{wrapfig}
\includeversion{report}
\excludeversion{conference}

\usepackage{slashbox}
\usepackage{fancybox}
\usepackage{amssymb}
\usepackage{amsfonts}
\usepackage{amsmath}
\usepackage{scalefnt}
\usepackage{url}
\usepackage{verbatim}
\usepackage{color}
\usepackage{subfig}

\usepackage{tikz-qtree}%

\usepackage{xspace}
\usepackage[latin1]{inputenc}

\usepackage[ruled,vlined,linesnumbered]{algorithm2e}

\newcommand{\optionalproof}[1]{\begin{report}\begin{proof}#1\end{proof}\end{report}}

\newcommand{\set}[1]{\left\{#1\right\}}

\newcommand{\setof}[2]{\left\{#1\,\middle|\:#2\right\}}

\newcommand{\dN}{\mathbb{N}}

\renewcommand{\card}[1]{\mathrm{card}(#1)}

\newcommand{\compl}[1]{\overline{#1}}

\newcommand{\Lit}{\mathcal{L}}
\newcommand{\submin}{\mathrm{SubMin}}
\newcommand{\Aimplg}[2]{\mathfrak{I}_{#1}(#2)}
\newcommand{\Aimpl}[1]{\Aimplg{\A}{#1}}
\newcommand{\PAimplg}[2]{\mathfrak{P}_{#1}(#2)}
\newcommand{\PAimpl}[1]{\PAimplg{\A}{#1}}
\newcommand{\partpi}[3]{\mathfrak{P}_{#1,#2}(#3)}
\newcommand{\partimp}[3]{\mathfrak{I}_{#1,#2}(#3)}
\newcommand{\basalgo}{\textsc{bp}}

\newcommand{\ucons}[2]{\mathrm{U}_{#1}(#2)}
\newcommand{\uprop}[3]{\mathrm{G}_{#1,#2,M}(#3)}
\newcommand{\impid}{\textsc{imp}}

\newcommand{\gpid}{\textsc{GPiD}\xspace}
\newcommand{\csp}{\textsc{cSP}\xspace}
\newcommand{\minisat}{\textsc{MiniSAT}\xspace}
\newcommand{\cvcfsmt}{\textsc{CVC4}\xspace}
\newcommand{\ztsmt}{\textsc{Z3}\xspace}

\newcommand{\gsub}{\trianglelefteq_{\theory}}
\newcommand{\eqsub}{\sim}
\newcommand{\subsetlit}[2]{#1[#2]}
\newcommand{\intsubsetlit}[3]{#1^#3[#2]}

\newcommand{\modelst}{\models_{\theory}}
\newcommand{\equivt}{\equiv_{\theory}}

\newcommand{\sord}{\prec}
\newcommand{\sordeq}{\preceq}
\newcommand{\UNSAT}[1]{$#1$ is \tunsat}
\newcommand{\timplicate}{$(\theory,\A)$-implicate\xspace}

\begin{document}

\title{A Generic Framework for Implicate Generation Modulo Theories}
\author{Mnacho Echenim \and Nicolas Peltier \and Yanis Sellami}
\institute{Univ. Grenoble Alpes, CNRS, LIG, F-38000 Grenoble France \email{[Mnacho.Echenim|Nicolas.Peltier|Yanis.Sellami]@univ-grenoble-alpes.fr}}

\newcommand{\Funs}{\Sigma}
\newcommand{\theory}{{\cal T}}

\newcommand{\formA}{\phi}

\newcommand{\tsat}{$\theory$-satisfiable\xspace}
\newcommand{\tunsat}{$\theory$-unsatisfiable\xspace}

\maketitle

\begin{abstract}
The clausal logical consequences of a formula are called its implicates. The generation of these implicates has several applications, such as the identification of missing hypotheses in a logical specification. We present a procedure that generates the implicates of a quantifier-free formula modulo a theory. No assumption is made on the considered theory, other than the existence of a decision procedure. The algorithm has been implemented (using the solvers \minisat, \cvcfsmt and \ztsmt) and experimental results show evidence of the practical relevance of the proposed approach.
\end{abstract}

\section{Introduction}

We present a novel approach based on the usage of a generic SMT solver as a black box to generate ground implicates of a formula
modulo a theory.
 Formally, the implicates of a formula $\phi$ modulo a theory $\T$ are the ground clauses $C$ such that every model of $\T$ that satisfies $\phi$ also satisfies $C$; in other words, these are the \emph{clausal $\T$-consequences of $\phi$}.
The problem of generating such implicates (up to logical entailment) is of great practical relevance, since for any implicate $\bigvee_{i=1}^n l_i$, the formula $\bigwedge_{i=1}^n \neg l_i \wedge \phi$ 
is \tunsat. The set $\{ \neg l_i \mid i \in [1,n] \}$ can thus be viewed as a set of hypotheses under which $\phi$ is \tunsat or, dually, $\neg \phi$ is provable.
This means that generating implicates can permit to identify missing hypothesis in a theorem, such as omitted lemmata or side conditions.
Such hypotheses are useful to correct mistakes in specifications, but also to quickly spot why a given statement is not provable.
They can be far more informative than counter-examples in this respect, since the latter are hard to analyze and can be 
clouded with superfluous information.

Consider for example the simple program over an array defined in Algorithm \ref{alg:example}.
\begin{algorithm}
    \caption{\texttt{Example}(Array[Int] $T$, Int $a$, Int $b$)}
    \label{alg:example}

    \SetKw{Requires}{requires}
    \SetKw{Let}{let}
    \SetKw{Ensures}{ensures}
    \Requires{$\forall x,y \in [a,b], x \leq y \implies T[x] \leq T[y]$}\;
    \Requires{$T[a] \geq 0$}\;
    \Let $T[b+1] = T[b-1] + T[b]$ \;
    \Ensures{$\forall x,y \in [a,b+1], x \leq y \implies T[x] \leq T[y]$}\;
\end{algorithm}
It turns out that the postcondition of the program is not verified.
This can be evidenced by translating the preconditions, the algorithm and the negation of the post-condition into a conjunction of logical formulas, and  using an SMT solver
to construct a model for this conjunction; this model can then be analyzed to determine what precondition is missing.
The obtained model, however, will generally contain a hard to read array definition, and the missing precondition will not be explicitly returned.
For instance, the model returned by the \ztsmt SMT solver  \cite{DBLP:conf/tacas/MouraB08}  
is (using our notations):
\begin{quote} 
	$a:533,\ b:533,\\
	 f:x \mapsto x \geq 533\ ?\ (x \geq 534\ ?\ 534 : 533) : 532,\\ 
	 g : x \mapsto x = 533\ ?\ 535 : (x = 534\ ?\ 19 : - 516),\\ 
	 T : x \mapsto  g(f(x))$.
\end{quote}
Implicate generation on the other hand permits to identify the missing precondition in a more efficient manner. The first step consists in selecting the literals that can be used to generate potential explanations; these are called \emph{abducible literals}. In this example, the natural literals to consider
are all the (negations of) equalities and inequalities constructed using constants $a$ and $b$, along with additional predefined constants such as $0$ and $1$. The second step simply consists in invoking our system, {\gpid},  to generate the potential missing preconditions. For this example, {\gpid} plugged with {\ztsmt} generates the missing precondition $a\neq b$ in less than  $0.2$ seconds. If abducible literals can be constructed using also the function symbol $T$, then our tool generates the other potential precondition $T[b-1] \geq 0$ in the same amount of time.

In previous work \cite{EP15,EPT13,EPT14,EPT17}, we devised refinements of the superposition calculus specially tuned to derive such implicates for 
quantifier-free formula modulo equality with uninterpreted function symbols. We proved the soundness and deductive-completeness of the obtained procedures, i.e., we showed that the procedure derives
all implicates up to redundancy.
In the present work, we investigate a different approach.
We provide a generic algorithm for generating such implicates, relying only on the existence of a decision 
procedure for the underlying theory, possibly augmented with counter-example generation capabilities to further restrict the search space. The main advantage of this approach is that 
it is possible to 
use efficient SMT solvers as black boxes, instead of having to develop specific systems for the purpose of implicate generation.
Our method is based on decomposition, in the spirit of the DPLL approach.
The generated implicates are constructed  on a given set of candidate literals, called {\em abducible} literals, which is assumed to be fixed before the beginning of the search, e.g., by a human user.
As far as flexibility is concerned, the algorithm also permits to only generate implicates satisfying so-called \emph{$\subseteq$-closed} predicates without any post-processing step. 
We show that the algorithm is sound and complete, and we provide experimental results showing that the obtained system is much more efficient than 
the previous one based on superposition.
We also devise generic approaches to store sets of implicates efficiently, while removing implicates that are redundant modulo the considered theory. 
Again, the proposed procedure relies only on the possibility of deciding validity in the underlying theory. 

\paragraph{Related work.} The implicate generation problem  has been thoroughfully investigated in the context of propositional logic (see for instance \cite{marquis2000consequence}).
Earlier approaches are based mainly  on refinements of the Resolution rule  \cite{jackson1992computing,kean1990incremental,quine1955way,tison1967generalization}, and they
focus on the definition of efficient strategies to generate saturated clause sets and of compact data structures for storing the generated sets of implicates \cite{deKleer1992improved,fredkin1960trie,Mishchenko01anintroduction,simon2001efficient}.
Other approaches use decomposition-based methods, in the style of the DPLL procedure, for generating trie-based representations of sets of prime implicates \cite{matusiewicz2009prime,matusiewicz2011tri}.
Recently \cite{previti2015prime}, a new approach that outperforms previous algorithms has been proposed, based on max-satisfiability solving and problem reformulation.
Our algorithm can be used for propositional implicate generation but it is not competitive with this new approach.  Our aim with this work was rather to extend the scope of implicate generation to more expressive logics.
Indeed, there have been only very few approaches dealing with logics other than propositional.
Some extensions have been considered in modal logics \cite{bienvenu2007prime,blackburn2007handbook}, and  algorithms have been proposed for  first-order formulas, based on first-order resolution \cite{knill1993equality,marquis1991extending} or tableaux \cite{mayer1993first,nabeshima2010solar}. However, none of these approaches is capable of handling equality efficiently.
More recently, algorithms were  devised to generate sets of implicants of formulas interpreted in decidable theories \cite{dillig2012minimum}, by combining quantifier-elimination for discarding useless variables, with model building to construct sufficient conditions for satisfiability.

The rest of the paper is structured as follows.
In Section \ref{sect:prel}, basic definitions and notations are introduced.
Section \ref{sect:gen} contains the definition of the algorithm 
for generating implicates, starting with a straightforward, naive algorithm
and refining it to make it more efficient.
In Section \ref{sect:store} data-structures and algorithms are presented to store implicates efficiently modulo redundancy.
Section \ref{sect:exp} contains the description of the implementation and experimental results, and Section \ref{sect:conc}
concludes the paper.
\begin{conference}
Due to space restrictions, some of the proofs are omitted. The full version is available on arXiv.
\end{conference}

\section{Preliminary notions}

\label{sect:prel}

Ground terms and non-quantified formulas are built inductively as usual on a sorted signature $\Funs$. 
The notions of validity, models, satisfiability, etc. are defined as usual.
The set of literals built on $\Funs$ is denoted by $\Lit$. 
Let $\theory$ be a theory. A set of formulas $S$ is {\em \tsat} if there exists an interpretation $I$ such that $I \models S$ and $I \models \theory$.
We assume that the $\theory$-satisfiability problem is decidable, i.e., that there exists an 
SMT solver that, given a formula $\formA$ with no quantifier, can decide whether $\formA$ is \tsat.

We consider clauses as unordered
disjunctions of literals with no repetition. 
Thus, when we write $C\vee D$, we implicitly
assume that $C$ and $D$ share no literal. We also identify unit clauses with the literal they contain. For every literal $l$, $\compl{l}$ denotes the literal complementary of $l$. The empty clause is denoted by $\false$. If $Q = \set{l_1, \ldots, l_n}$ is a set of literals, then we denote by $\compl{Q}$ the clause $\compl{l_1} \vee \cdots \vee \compl{l_n}$. Conversely, given a clause $C = l_1\vee \cdots \vee l_n$, we denote by $\compl{C}$ the set of literals (or unit clauses) $\set{\compl{l_1},\ldots, \compl{l_n}}$.

We consider a finite set of \emph{abducible literals} $\A$. We assume that each of these literals is \tsat. Given a set of clauses $S$, we call a clause $C$ a \emph{\timplicate of $S$} if $\compl{C} \subseteq \A$
and $S\modelst C$. We say that $C$ is a \emph{prime \timplicate of $S$} if $C$ is a \timplicate of $S$ and for every \timplicate $D$ of $S$, if $D\modelst C$ then $C\modelst D$. The set of {\timplicate}s of $S$ is denoted by $\Aimpl{S}$, and the set of prime {\timplicate}s of $S$ is denoted by $\PAimpl{S}$.

\newcommand{\gsrep}[2]{\Theta_{#1}(#2)}
Given a set of clauses $S$ and a clause $C$, we write $S\gsub C$ if there is a clause $D\in S$ such that $D\modelst C$.
If $S'$ is  a set of clauses, then we write $S\gsub S'$ if for all $C \in S'$, we have $S\gsub C$.
We write $S \eqsub S'$ if $S \gsub S'$ and $S' \gsub S$ (i.e., $S$ and $S'$ are identical modulo $\theory$-equivalence).

\begin{proposition}
\label{prop:model}
Let $l$ be a literal and let $C,D$ be clauses.
The following statements hold:
\begin{enumerate}
\item{$l \vee C \modelst D$ iff $l \modelst D$ and $C \modelst D$.\label{model:decomp}}
\item{$C \modelst l \vee D$ iff $C \wedge  \compl{l} \modelst D$. \label{model:shift}}
\end{enumerate}
\end{proposition}

\newcommand{\tautology}{$\theory$-tautology\xspace}
\newcommand{\tautologies}{$\theory$-tautologies\xspace}

We assume an order $\sord$ is given on clauses built on $\A$ that agrees with inclusion, i.e., such that $C \subsetneq D \Rightarrow C \sord D$.

\begin{definition}
A {\em \tautology} is a clause that is satisfied by every model of $\theory$.
Given a set of clauses $S$, we denote by $\submin(S)$ the set obtained by deleting from $S$ all clauses $D$ such that either $D$ is a \tautology, or there exists $C \in S$ such that $C \modelst D$ and ($D \not \modelst C$ or $C \sord D$).
\end{definition}

Note that in particular, we have $\PAimpl{S} \eqsub \submin(\Aimpl{S})$.

\newcommand{\fixed}{$\theory$-fixed\xspace}
\newcommand{\free}{$\theory$-free\xspace}
\newcommand{\depends}{$\theory$-depends\xspace}

\newcommand{\intof}[2]{#1[#2]}
\newcommand{\notof}[1]{\neg}

\section{On the generation of prime {\timplicate}s}

\label{sect:gen}

\subsection{A basic algorithm}

We present a simple and intuitive algorithm that permits to generate the {\timplicate}s of a set of formulas $S$. This algorithm is based on the fact that a clause $C$ is a \timplicate of $S$ if and only if $\compl{C} \subseteq \A$ and $S\cup \compl{C} \modelst \false$. It will thus basically consist in enumerating the subsets of $\A$ and searching for those whose union with $S$ is \tunsat. This may be done by starting with an empty set of hypotheses $M$ and repeatedly and nondeterministically adding new abductible literals to $M$ until $S \cup M$ is \tunsat.
This algorithm is naive, as the same clauses will be produced multiple times, but it forms the basis of the more efficient algorithm in Section \ref{sect:improve}.

\begin{definition}\label{def:partimp}
Let $S$ be a set of formulas.
Let $M,A$ be sets of literals such that $M\cup A \subseteq \A$. We define
	\begin{eqnarray*}
	\partimp{M}{A}{S}& =& \setof{C \in \Aimpl{S}}{\exists Q \subseteq A,\, C = \compl{M} \vee \compl{Q}
},\\
	\partpi{M}{A}{S} & = & \submin(\partimp{M}{A}{S}).
	\end{eqnarray*}
\end{definition}
Intuitively, a clause $\compl{M}\vee \compl{Q}$ thus belongs to $\partimp{M}{A}{S}$ if and only if $\compl{Q}$ is a \timplicate of $S\cup M$.

\begin{proposition}
\label{prop:unsatM}
Let $S$ be a set of formulas, $M$ and $A$ be sets of literals such that $M \cup A \subseteq \A$. If $M$ is \tsat, then
$S \cup M$ is \tunsat iff $\partpi{M}{A}{S} = \{ \compl{M} \}$.
If $M$ is \tunsat, then $\partpi{M}{A}{S} = \emptyset$.
\end{proposition}
\optionalproof{
If $S \cup M$ is \tunsat then $S \modelst \compl{M}$ and $\compl{M} \in \Aimpl{S}$, thus $\compl{M} \in \partimp{M}{A}{S}$ (by letting $Q = \emptyset$ in Definition \ref{def:partimp}). By definition, for any clause $C \in \partimp{M}{A}{S}$, we have $\compl{M} \subseteq C$; thus $\compl{M} \sordeq C$ and $\compl{M} \modelst C$. Since $\compl'M$ is not a {\tautology} by hypothesis, we deduce that $\submin(\partpi{M}{A}{S}) = \{ \compl{M} \}$. 
Conversely, if $\partpi{M}{A}{A} = \{ \compl{M} \}$ then $\compl{M} \in \Aimpl{S}$ by definition, hence $S \cup M$ is \tunsat.

If $M$ is \tunsat, then any clause containing $\compl{M}$ is a \tautology. Consequently, all clauses in $\partimp{M}{A}{S}$ are \tautologies, and $\partpi{M}{A}{S}$ is empty.
}

\begin{report}
\begin{proposition}
\label{prop:evalequiv}
Let $S$ be a set of formulas and $M \subseteq \A$.
If $S \cup M \equivt S' \cup M$, then $\partimp{M}{A}{S} = \partimp{M}{A}{S'}$.
\end{proposition}
\optionalproof{
Assume that $C \in \partimp{M}{A}{S}$. Then $C = \compl{M} \vee \compl{Q}$, with $Q \subseteq A$ and $C\in \Aimpl{S}$.
Since $C \in \Aimpl{S}$ we have $S \models C$, hence $S \cup M \cup Q \modelst \false$.
Since $S \cup M \equivt S' \cup M$ we deduce that $S' \cup M \cup Q \modelst \false$, hence $C \in \Aimpl{S'}$, and therefore $C \in \partimp{M}{A}{S'}$. 
Consequently,  $\partimp{M}{A}{S} \subseteq \partimp{M}{A}{S'}$. By symmetry, we deduce that $\partimp{M}{A}{S} = \partimp{M}{A}{S'}$.
}
\end{report}

\newcommand{\fixlit}[3]{\mathtt{fix}(#3,#2,#1)}

It is clear that it is useless to add a new hypothesis $l$ into $M$ both if $M \cup \{ l \}$ is \tunsat (because the obtained \timplicate would be a \tautology), or if this set is equivalent to $M$ (because the \timplicate would not be minimal). This motivates the following definition:
\begin{definition}
Let $S$ be a set of formulas and let $M,A$ be two sets of literals.
We denote by $\fixlit{A}{M}{S}$ a set
obtained by deleting from $A$ some literals
 $l$ such that either $M \cup S \modelst l$ or $M \modelst \compl{l}$.
\end{definition}
The use of this definition aims to reduce the number of abducible hypotheses to try, and thus the search space of the algorithm.
Still, we do not assume that all the literals $l$ satisfying the condition above are deleted because, in practice, such literals may be hard to detect.
However, we assume that no element from $M$ is in $\fixlit{A}{M}{S}$.
\begin{proposition}\label{prop:partpi_incr}
	Consider a set of formulas $S$ and two sets of literals $M,A$ such that $\partpi{M}{A}{S} \neq \set{\compl{M}}$. The following equalities hold:
	\begin{enumerate}
		\item $\partpi{M}{A}{S} = \submin(\bigcup_{l\in A
} \partpi{M\cup \set{l}}{A}{S})$.\label{it:union}
		\item $\partpi{M}{A}{S} = \partpi{M}{\fixlit{A}{M}{S}}{S}$.\label{it:unit}
	\end{enumerate}
\end{proposition}
\optionalproof{
\hfill
\begin{enumerate}
\item{
It suffices to prove that $\partimp{M}{A}{S} = \bigcup_{l \in A \setminus M} \partimp{M\cup \set{l}}{A}{S})$.
	Let $C\in \partimp{M}{A}{S}$. By hypothesis, $C$ is of the form $\compl{M} \vee \compl{Q}$, where $Q \subseteq A$ and $M \cap Q = \emptyset$. Since $\partpi{M}{A}{S} \neq \set{\compl{M}}$, necessarily $C \not = \compl{M}$ and $Q\neq \emptyset$. Let $l\in Q$ and $m = \compl{l}$.
	We have $C = \compl{M} \vee m \vee \compl{Q'}$, with $Q' = Q \setminus \set{l}$, and since $l\in A$, $C\in \partimp{M\cup \set{l}}{A}{S}$. 
Conversely, if $C \in \partimp{M \cup \set{l}}{A}{S}$ with $l \in A$, 
then $C \in \Aimpl{S}$ and $C = \compl{M} \vee \compl{l}  \vee Q$, for some $Q \subseteq A$, so that $C \in \partimp{M}{A}{S}$.}
\item{Since $\fixlit{A}{M}{S}\subseteq A$, we have $\partpi{M}{\fixlit{A}{M}{S}}{S}\subseteq \partpi{M}{A}{S}$. Now let $m$ be a literal in $A$ such that either $M \cup S \models m$ or $M \models \compl{m}$. 
If $C\in \submin(\partimp{M}{A}{S})$ is of the form $\compl{M} \vee \compl{Q}$, then $Q$ cannot contain $m$: indeed, in the former case  $\compl{m}$ could be removed from $S \cup \compl{C}$ while preserving equivalence, hence the implicate would not be minimal, and in the later case $C$ would be a \tautology.}
\end{enumerate}
}
The results above lead to a basic algorithm for generating {\timplicate}s which is described in Algorithm \ref{alg:basalgo}. As explained above, the algorithm works by adding literals from $A$ as hypotheses until a contradiction can be derived. The \textbf{return} statement at Line \ref{line:basret} avoids enumerating the subsets that contain $M$, once it is known that $S\cup M$ is \tunsat.
 \begin{algorithm}
 	\caption{\basalgo($S, M, A$)}\label{alg:basalgo}
 	\SetKw{Let}{let}
 	\SetKw{Return}{return}
 	\If{\UNSAT{M}}
 	{\Return $\emptyset$\;\label{line:unsatret}}
 	\Else{
 	\If{\UNSAT{S\cup M}}
{\Return $\set{\compl{M}}$\;\label{line:basret}}
 	\Else{$B = \fixlit{A}{M}{S}$\;
 		\ForEach{$l \in B$  }
 		{\Let{$P_l = \basalgo(S, M\cup \set{l}, B)$}\label{line:basa}\;	 		
 		}
	 	\Return $\submin(\bigcup_{l\in B}P_l)$\;
	 }}
 \end{algorithm}

\begin{lemma}
	$\partpi{M}{A}{S} = \basalgo(S, M, A)$.
\end{lemma}
\begin{proof}
	The result is proved by  a straightforward induction on $\card{\A\setminus M}$. By Proposition \ref{prop:unsatM}, $\partpi{M}{A}{S} = \emptyset$ if $M$ is \tunsat,
    and
	$\partpi{M}{A}{S} = \set{\compl{M}}$ if $M$ is \tsat and $S \cup M$ is \tunsat.
    Otherwise, by Proposition \ref{prop:partpi_incr}(\ref{it:unit}), we have $\partpi{M}{A}{S} = \partpi{M}{\fixlit{A}{M}{S}}{S} = \partpi{M}{B}{S}$.
	By the induction hypothesis, for each $l \in A$, $P_l = \partpi{M\cup \set{l}}{B}{S}$, and
 by Proposition \ref{prop:partpi_incr}(\ref{it:union}), $P_l = \partpi{M\cup \set{l}}{A}{S}$; 
 we deduce that $\partpi{M}{A}{S} = \submin(\bigcup_{l\in A}P_l)$.
 Note that at each recursive call, a new element is added to $M$, since $\fixlit{A}{M}{S}$ is assumed not to contain any element from $M$.
\end{proof}
\begin{theorem}
If $S$ is a set of formulas then
	$\PAimpl{S} = \basalgo(S, \emptyset, \A)$.
\end{theorem}
Although the algorithm described above computes all the prime {\timplicate}s of any clause set as required, it is very inefficient, in particular because of the large number of useless and redundant recursive calls that are made.
In what follows we present several improvements to the algorithm in order to generate implicates as efficiently as possible.

\subsection{Restricting the set of candidate hypotheses}

\label{sect:improve}

It is obvious that the algorithm {\basalgo} makes a lot of redundant calls: for example, if $l_1\vee l_2$ is a prime \timplicate of a clause set $S$, then this \timplicate will be generated twice, first as $l_1\vee l_2$, and then as $l_2\vee l_1$. Such redundant calls are quite straightforward to avoid by ensuring that every invocation of the algorithm contains a distinct set of literals $M$. This can be done by fixing an ordering $<$ among literals in $\A$, and by assuming that hypotheses are always added in this order.
Another way of restricting the set of candidate hypotheses is to exploit  information extracted from the previous satisfiability test.
For example, if $S\cup\set{l_1}$ is satisfiable for some literal $l_1$, and that a model of this set satisfies another literal $l_2$, then $S\cup \set{l_2}$ is also satisfiable and it is not necessary to consider $l_2$ as a hypothesis.
In particular, if a model of $S\cup \set{l_1}$ validates all the literals in $A$, then $\partpi{M}{A}{S}$ is necessarily empty and no literal should be selected.
We can thus take advantage of the existence of a model of $S \cup M$ in order to guide the choice of the next literals in $A$. 
However, observe that this refinement interferes with the previous one based on the order $<$. Indeed, non-minimal hypotheses will have to be considered if all the smaller hypotheses  are dismissed because they are true in the model. We  formalize these principles below.

\begin{definition}
\label{def:subsetlit}
	In what follows, we consider a total ordering\footnote{Note that this ordering is not necessarily related to the ordering $\sord$ on clauses.} $<$ on the elements of $\A$.
For $A \subseteq \A$ and $l\in A$, we define $\subsetlit{A}{l} \isdef \setof{l' \in A}{l < l'}$.
If $I$ is a set of literals then
we denote by $\intsubsetlit{A}{l}{I}$ the set
$\setof{l' \in A}{l' < l \wedge \compl{l'} \not \in I} \cup \subsetlit{A}{l}$.
\end{definition}
\begin{example}
	Assume that $A = \setof{p_i, \neg p_i}{i = 1, \ldots, 6}$ and that for all literals $l \in \set{p_i, \neg p_i}$ and $l' \in \set{p_j, \neg p_j}$, $l < l'$ if and only if either $i < j$ or ($i = j$, $l = p_i$ and $l' = \neg p_i$). Then $A[p_4] = \set{\neg p_4, p_5, \neg p_5, p_6, \neg p_6}$.	
	If $I = \set{p_1, \neg p_2}$, then $\intsubsetlit{A}{p_4}{I} = \set{p_1,  \neg p_2, p_3, \neg p_3, \neg p_4, p_5, \neg p_5, p_6, \neg p_6}$.
\end{example}

\newcommand{\scompat}[1]{{#1}-\textrm{compatible}}
\begin{definition}\label{def:scompat}
	Let $S$ be a set of clauses. A set of literals $I$ is \emph{\scompat{$S$} with respect to $\A$} (or simply \emph{\scompat{$S$}}) if every prime \timplicate of $S$ contains a literal $l$ such that $l \in I$.
\end{definition}

Intuitively, an \scompat{$S$} set $I$ consists of literals $l$ such that $\compl{l}$ will be allowed to be added as a hypothesis to generate {\timplicate}s of $S$ (see Lemma \ref{lm:ucm} below). 
The set $I$ can always be defined by taking the negations of all the abducible literals from $\A$. In this case, all literals will remain possible hypotheses. It is possible, however, to restrict the size of $I$ when a
model of $S$ is known, as evidenced by the following proposition:

\begin{proposition}\label{prop:scompmodel}
	If $S$ is a set of clauses and $J$ is a 
	model of $S$, then the set $I\isdef \setof{l \in \Lit}{J\models l}$ is \scompat{$S$}.
\end{proposition}

\begin{proof}
	Let $Q$ be a set of literals such that $\compl{Q}$ is a prime \timplicate of $S$, and assume that for all $l\in Q$, $\compl{l} \not \in I$, i.e., that for all $l\in Q$, $J\not \models \compl{l}$. Then $J \models l$ holds for every $l \in Q$, hence 
	$J\models S\cup Q$ and $\compl{Q}$ cannot be
	a \timplicate of $S$.
\end{proof}
Note that the condition of having a
model of $S$ was not added to Definition \ref{def:scompat} because in practice, such a model cannot always be constructed efficiently.

\newcommand{\evaluate}[2]{#1_{#2,M}}

 Being able to derive unit consequences of the set of axioms (for instance by using unit propagation), can pay off if this additional information can be used to simplify the formula at hand. This motivates the following definition.
\begin{definition}
\label{def:su}
Let $S$ be a set of formulas and $M \subseteq \A$.
	We denote by $\ucons{M}{S}$ the set of unit clauses logically entailed by $S\cup M$ modulo $\theory$, i.e., $\ucons{M}{S} \isdef \setof{l \in \Lit}{S\cup M\modelst l}$. Given a set $U$ such that $M\subseteq U\subseteq \ucons{M}{S}$, we denote by $\evaluate{S}{U}$ the formula obtained from $S$ by replacing {\em some} (arbitrarily chosen) literals $l'$
		by $\false$ if $U\modelst \compl{l'}$ 
		and by $\true$ if $M \modelst l'$.
\end{definition}
Note that $U$ is not necessarily identical to $\ucons{M}{S}$, because in practice the latter set is hard to generate. Similarly we do not assume that all literals $l'$ are replaced in Definition \ref{def:su}
since testing logical entailment may be costly. Lemma \ref{lem:evalimp} shows that the {\timplicate}s of a set $S$ and those of $\evaluate{S}{U}$ are identical.
\begin{report}
The proof uses the following proposition.
\begin{proposition}
\label{prop:impequiv}
Let $S$ be a set of formulas and $M \subseteq \A$.
Consider a set of literals $U$ such that $M\subseteq U\subseteq \ucons{M}{S}$.
Then $S\cup M \equivt \evaluate{S}{U}\cup M$. 
\end{proposition}
\begin{proof}
We have $S \cup M \models \evaluate{S}{U}$, since $S \cup M \modelst U$, and
 $U \models (l' \Leftrightarrow \false)$ if $U \models \compl{l'}$.
Conversely, it is clear that $\evaluate{S}{U} \cup M \models S \cup M$.
\end{proof}
\end{report}
\begin{lemma}
\label{lem:evalimp}
Let $S$ be a set of formulas and $M \subseteq \A$.
Consider a set of literals $U$ such that $M\subseteq U\subseteq \ucons{M}{S}$.
Then
$\partimp{M}{A}{S} = \partimp{M}{A}{\evaluate{S}{U}}$
\end{lemma}
\optionalproof{
This is an immediate consequence of Propositions \ref{prop:evalequiv} and \ref{prop:impequiv}.
}
\begin{definition}
Let $U,M,A$ be sets of literals. We define:
$\uprop{U}{A}{S} \isdef \setof{\compl{M}\vee \compl{l}}{l \in {A}\wedge \compl{l}\in  U}$. 
\end{definition}

The lemma below can be viewed as a refinement of Proposition \ref{prop:partpi_incr}. It is based on the previous results, according to which, when adding a new hypothesis $l$, it is possible to remove from the set of abducible literals $A$ every literal that is strictly smaller than $l$, provided its complementary is in $I$ (because we can always assume that the smallest available hypothesis is considered first).
This is why the recursive call is on  $\intsubsetlit{A}{l}{I}$ instead of $A$. 
Note also that the use of semantic guidance interferes with the use of the ordering $<$: the smaller the set $I$, the larger $\intsubsetlit{A}{l}{I}$.

\begin{lemma}\label{lm:ucm}
	Assume that $S \cup M$ is \tsat and
	let $I$ be an \scompat{$(S\cup M)$} set of literals.
Let $U$ be a set of literals such that $M\subseteq U\subseteq \ucons{M}{S}$.
We have
 \[\partpi{M}{A}{S}\ =\ \submin\left(\uprop{U}{A}{S} \cup \bigcup_{l\in A, \compl{l} \in I} \partpi{M \cup \set{l}}{\intsubsetlit{A}{l}{I}}{S}\right).\]
\end{lemma}

\begin{proof}
	First note that $\partpi{M}{A}{S} \neq \set{\compl{M}}$, since $S \cup M$ is \tsat.
	We first prove that $\partpi{M}{A}{S} \subseteq \uprop{U}{A}{S} \cup \bigcup_{l\in A, \compl{l} \in I} \partpi{M \cup \set{l}}{\intsubsetlit{A}{l}{I}}{S}$.
	Let $C \in \partpi{M}{A}{S}$. 
	By hypothesis, $C$ is of the form $\compl{M}\vee \compl{Q}$, where $\emptyset\neq Q  \subseteq A$.
Since $I$ is \scompat{$(S\cup M)$}, $Q$ necessarily contains a literal $l\in A$ such that $\compl{l} \in I$. 
Assume that $l$ is the smallest literal in $Q$ satisfying this property.
We distinguish the following cases.
	\begin{description}
		\item Assume that $Q$ contains a literal $l'$ such that $\compl{l'}\in U$.
In this case, since $U \subseteq \ucons{M}{S}$, $S \cup M \modelst \compl{l'}$. Since $Q \subseteq A$, we also have $l' \in A$, and since $\partpi{M}{A}{S} \neq \set{\compl{M}}$,  we deduce that $\compl{M} \vee \compl{l'} \in \partpi{M}{A}{S}$. Since $\compl{M} \vee \compl{l'} \modelst C$ and $C \in \partpi{M}{A}{S}$, $C$ must be smaller or equal to $\compl{M} \vee \compl{l'}$, which is possible only if $C = \compl{M} \vee \compl{l'}$. 
We deduce that $C\in \uprop{U}{A}{S}$.
		\item  Otherwise,
we show that $Q \setminus \set{l} \subseteq \intsubsetlit{A}{l}{I}$. 
By Definition \ref{def:subsetlit}, we have $\subsetlit{A}{l} = \setof{l' \in A}{l < l'}$ and $\intsubsetlit{A}{l}{I} = \setof{l' \in A}{l' < l \wedge l' \not \in I} \cup \subsetlit{A}{l}$.
Let $l' \in Q$, with $l' \not = l$.
If $l' > l$ then $l' \in \subsetlit{A}{l} \subseteq \intsubsetlit{A}{l}{I}$. If $l' \not > l$, then since $>$ is total and $l \not = l'$, necessarily $l > l'$.
Since $l$ is the smallest literal in $Q$ such that $\compl{l} \in I$, we deduce that $\compl{l'} \not \in I$.  Thus $l' < l$ and $\compl{l'}\not \in I$, which entails that $l' \in \intsubsetlit{A}{l}{I}$.
		Consequently, $Q \setminus \set{l} \subseteq \intsubsetlit{A}{l}{I}$. Since $C = \compl{M \cup \set{l}} \vee \compl{Q \setminus \set{l}}$, this entails that $C \in \partpi{M\cup \set{l}}{\intsubsetlit{A}{l}{I}}{S}$.
	\end{description}	
		We now prove that $\uprop{U}{A}{S} \cup \bigcup_{l\in B, \compl{l} \in I} \partpi{M \cup \set{l}}{\intsubsetlit{B}{l}{I}}{S} \subseteq \partimp{M}{A}{S}$.

	\begin{description}
	\item{
	Let $C \in \uprop{U}{A}{S}$.
	By definition, $C$ is of the form $\compl{M} \cup l$ with $l \in \compl{A} \cap U$.
	Since $U \subseteq \ucons{M}{S}$, we deduce that $S \cup M \modelst l$, i.e., that $S \modelst \compl{M} \vee l$.
	Since $\compl{l} \in A$, this entails that $\compl{M} \vee l \in \partimp{M}{A}{S}$, hence $C \in \partimp{M}{A}{S}$.
	}
	\item{Let $C \in \partpi{M \cup \set{l}}{\intsubsetlit{A}{l}{I}}{S}$ with $l \in A$, $\compl{l} \in I$.
	By definition, $C = \compl{M} \vee \compl{l} \vee \compl{Q}$, with $Q \subseteq \intsubsetlit{A}{l}{I}$
	and $C \in \Aimpl{S}$. 
	But $\intsubsetlit{A}{l}{I} \subseteq A$ by definition, thus $Q \cup \set{l} \subseteq A$ and $C = M \vee (\compl{Q} \vee \compl{l}) \in \partimp{M}{A}{S}$.
	}
	\end{description}
\end{proof}

\newcommand{\closed}{$\subseteq$-closed}

Similarly to {\csp} (see \cite[Sect. 4.2]{EPT17}), we parameterize our algorithm by a predicate in order to filter the implicates that are generated.
The goal of this parametrization is to allow the user to restrict the form of the generated implicates.
Typically, one could want to generate implicates only up to a given size limit, or only those satisfying some specific semantic constraints.
\begin{definition}
	A predicate $\calP$ on sets of literals is \emph{\closed} if for all sets of literals $A$ such that $\calP(A)$ holds, if $B\subseteq A$ then $\calP(B)$ also holds.
\end{definition}

Examples of {\closed} predicates include cardinality constraints: $\calP_k \isdef \lambda A.\ \card{A} \leq k$, where $k\in \dN$, or implicant constraints: $\calP_\phi \isdef \lambda A.\ \phi \models A$, where $\phi$ is a formula. Note that {\closed} predicates can safely be combined by the conjunction and disjunction operators.

An important feature of {\closed} predicates is that implicates verifying such predicates can be generated on the fly without any post-processing step, thanks to the following result:
\begin{proposition}\label{prop:closed}
	If $\calP$ is {\closed} and $\calP(M)$ does not hold, then for all sets of literals $A$, $\calP(M\cup A)$ does not hold either.
\end{proposition}

The inclusion of these improvements to the original algorithm results in the one described in Algorithm \ref{alg:impid}.

\begin{algorithm}
	\caption{\impid($S, M, A, \calP$)}\label{alg:impid}
	\SetKw{Let}{let}
	\SetKw{And}{and}
	\SetKw{Or}{or}
\If{\UNSAT{M} \Or $\neg \calP(M)$}
	{
		\Return $\emptyset$\;	
	}
	\If{\UNSAT{S \cup M}}
	{		
		\Return $\set{\compl{M}}$\;
	}
	\Let{$U \subseteq \ucons{M}{S}$ such that $M\subseteq U$}  \label{line:propsimp}\;
	\Let{$S = \evaluate{S}{U}$}  \label{line:psetsimp}\;
	\Let{$A = \fixlit{A}{M}{S}$}  \label{line:linkedsimp}\;
	\Let{$I$ be an \scompat{$(S\cup M)$} set of literals} \label{line:modelsimp}\;
	\ForEach{$l\in A$ such that $\compl{l} \in I$}
	{		
		\Let{$P_l = \impid(S, M\cup \set{l}, \intsubsetlit{A}{l}{I}, \calP)$\;}		
	}
	\Return $\submin(\uprop{U}{A}{S}\cup \bigcup_{l\in A} P_l)$\;	
\end{algorithm}

\begin{lemma}
	If $\calP$ is {\closed} then
	$\impid(S,M,A,\calP) = \partpi{M}{A}{S} \cap \setof{\compl{A}}{A\in \calP}$.
\end{lemma}

\begin{proof}
	If one of $M$ or  $S\cup M$ is \tunsat, or $\calP(M)$ does not hold, then the result follows from Propositions \ref{prop:unsatM} and \ref{prop:closed}. Otherwise the result is proved by induction on $\card{\A \setminus M}$, using Proposition \ref{prop:partpi_incr} and Lemmata \ref{lm:ucm} and \ref{lem:evalimp}.
\end{proof}

\begin{theorem}
	If $\calP$ is {\closed} then $\PAimpl{S}\cap \setof{\compl{A}}{A\in \calP} = \impid(S, \emptyset, \A, \calP)$.
\end{theorem}

\section{On the storage of {\timplicate}s}

\label{sect:store}

\newcommand{\trie}{$\A$-tree\xspace}
\newcommand{\falsetrie}{\bot}
\newcommand{\clset}[1]{{\cal S}(#1)}

\newcommand{\atreen}{\{\apair{l_1}{\tau_1}, \dots, \apair{l_n}{\tau_n}\}}
\newcommand{\atreeempty}{\{ \apair{l}{\{\apair{l_1}{\emptyset}, \dots, \apair{l_n}{\emptyset}\}} \}}
\newcommand{\apair}[2]{#1\!:#2}

\newcommand{\delete}{\mathtt{Simp}}

\newcommand{\prune}[2]{\mathtt{rm}(#1,#2)}

\newcommand{\ordt}{<_t}

The number of implicates of a given formula may be huge, 
hence it is essential in practice to have appropriate data structures 
to store them in a compact way and efficient algorithms to check  
that a newly generated implicate $C$ is not redundant 
(forward subsumption modulo $\theory$), and if so, to
delete all the already generated implicates that are less general than $C$ (backward subsumption modulo $\theory$), 
before inserting 
$C$ into the stored implicates.
In this section, we devise a trie-like data-structure to perform these tasks.
As in the previous section, we only rely on the existence of a
decision procedure for testing $\theory$-satisfiability.

\begin{definition}
Let $\ordt$ be an order on the literals in $\A$, possibly, but not necessarily, equal to the order $<$ used for literal ordering in the implicate generation algorithm.
An {\em \trie} is inductively defined as $\falsetrie$ or a possibly empty set of pairs 
$\atreen$, where $l_1,\ldots,l_n$ are pairwise distinct literals in $\A$ and $\tau_i$ (for $i=1,\dots,n$) is an \trie only containing
literals that are strictly $\ordt$-greater than $l_i$.
An \trie is associated with a set of $\A$-clauses inductively defined as follows:
\[
\begin{tabular}{ccc}
$\clset{\falsetrie}$	& $\isdef$	& $\set{\false}$, \\
$\clset{\atreen}$			
& $\isdef$ 
& $\bigcup_{i=1}^n \{ l_i \vee C \mid C \in \clset{\tau_i}\}.$
\end{tabular}\]
\end{definition}
In particular, $\clset{\emptyset} = \emptyset$. Intuitively an \trie may be seen as a tree in which the edges are labeled by literals and the leaves are labeled by $\emptyset$ or $\falsetrie$, and represents a set of clauses corresponding to paths from the root to $\falsetrie$.
We introduce the following simplification rule (which may be applied at any depth inside a tree, not only at the root level):
\[
\begin{tabular}{cccc}
$\delete:$ & $\tau \cup \set{\apair{l}{\emptyset}}$ & $\rightarrow$ & $\tau$ \\
\end{tabular}
\]
Informally, the rule deletes all leaves labeled by $\emptyset$ except for the root. It may be applied recursively, for instance $\atreeempty \rightarrow_{\delete}^{n+1} \emptyset$. Termination is immediate since the size of the tree is strictly decreasing.
\begin{proposition}
If $\tau \rightarrow_{\delete} \tau'$ then  $\tau'$ is an  \trie and $\clset{\tau} = \clset{\tau'}$.
\end{proposition}

The algorithm permitting the insertion of a clause in an \trie is straightforward and thus omitted.
The following lemma provides a simple algorithm to check whether a clause is a logical consequence modulo $\theory$ of some clause in $\clset{\tau}$ (forward subsumption). The algorithm proceeds by induction on the \trie.
\begin{lemma}
\label{lem:forward}
Let $C$ be a clause and let $\tau$ be an \trie. 
We have $\clset{\tau} \gsub C$ iff one of the following conditions hold:
\begin{itemize}
\item{$\tau = \falsetrie$.}
\item{$\tau = \atreen$
and there exists $i \in [1,n]$
such that $l_i \modelst C$
and $\clset{\tau_i} \gsub C$.}
\end{itemize}
\end{lemma}
\begin{proof}
If $\tau = \falsetrie$ then $\clset{\tau} = \{ \false \} \gsub C$ hence the equivalence holds.
Otherwise, let $\tau = \atreen$. By definition, $\clset{\tau} \gsub C$ holds iff there exists a clause $D \in \clset{\tau}$ such that 
$D \modelst C$. Since $\clset{\atreen} = \bigcup_{i=1}^n \{ l_i \vee E \mid E \in \clset{\tau_i}\}$, the previously property holds iff 
there exists $i \in [1,n]$ and $E \in \clset{\tau_i}$ such that $l_i \vee E \modelst C$, i.e., such that
$l_i \modelst C$ and $E \modelst C$ by Proposition \ref{prop:model}(\ref{model:decomp}). 
By definition, $\exists E\, (E \modelst C \wedge E \in \clset{\tau_i})$ iff $(\clset{\tau_i} \gsub C)$. Furthermore, $l_i \gsub C$ holds iff
$\compl{C} \cup \set{l_i}$ is \tunsat, hence the result.
\end{proof}
The following definition provides an algorithm to remove, in a given \trie, all branches corresponding to clauses that are logical consequences of a given formula modulo $\theory$ (backward subsumption).
\begin{definition}
\label{def:backward}
Let $\formA$ be a formula and let $\tau$ be an \trie. 
$\prune{\tau}{\formA}$ denotes the \trie defined as follows:
\begin{itemize}
\item{If $\formA$ is \tunsat, then $\prune{\tau}{\formA} \isdef \emptyset$.}
\item{If $\formA$ is \tsat, then:
\begin{itemize}
\item{$\prune{\falsetrie}{\formA} \isdef \falsetrie$, }
\item{$\prune{\atreen}{\formA} \isdef 
\bigcup_{i=1}^n \{ \apair{l_i}{\prune{\tau_i}}{\formA \wedge \compl{l_i}} \}$.}
\end{itemize}}
\end{itemize}
\end{definition}
Intuitively, starting with some clause $C$, the algorithm incrementally adds literals $\compl{l_1},\dots,\compl{l_n}$ occurring in the clauses $D = l_1 \vee \dots \vee l_n \in \clset{\tau}$ and invokes the SMT solver after each addition. If a contradiction is found then this means that $C \modelst D$, hence the branch corresponding to $D$ can be removed. The calls are shared among all common prefixes.
Of course, this algorithm is interesting mainly if the SMT solver is able to perform incremental satisfiability testing, with ``push'' and ``pop'' commands to add and remove formulas from the set of axioms (which is usually the case). 
\begin{lemma}
Let $\formA$ be a formula and let $\tau$ be an \trie.
Then $\prune{\tau}{\formA}$ is an \trie, and $\clset{\prune{\tau}{\formA}} = \{ C \in \clset{\tau} \mid \formA \not \modelst C \}$.
\end{lemma}
\optionalproof{
The proof is by induction on $\tau$. We also prove that all literals in $\prune{\tau}{\formA}$ also occur in $\tau$. 
We distinguish several cases.
\begin{itemize}
\item{If $\formA$ is \tunsat, then $\formA \modelst C$, for every $C \in \clset{\tau}$. 
Consequently, $\{ C \in \clset{\tau} \mid \formA \not \modelst C \} = \emptyset$.
Furthermore, $\prune{\tau}{\formA} = \emptyset$, thus $\prune{\tau}{\formA}$ is an \trie and $\clset{\prune{\tau}{\formA}} = \emptyset$. Therefore, the result holds. 
Note that $\prune{\tau}{\formA}$ contains no literal at all.}
\item{If $\formA$ is \tsat and $\tau = \falsetrie$, then $\clset{\tau} = \{ \false \}$ and
$\formA \not \gsub \false$, thus $\{ C \in \clset{\tau} \mid \formA \not \modelst C \} = \{ \false \}$.
Furthermore, $\prune{\tau}{\formA} = \falsetrie$, thus $\prune{\tau}{\formA}$ is an \trie and the equality holds.
As in the previous case, $\prune{\tau}{\formA}$ contains no literal.
}
\item{Now assume $\formA$ is \tsat and $\tau = \atreen$, where $l_1,\dots,l_n$ are pairwise distinct, and all literals in $\tau_i$ are strictly greater than $l_i$. Then
$\clset{\atreen} = \bigcup_{i=1}^n \{ l_i \vee D \mid D \in \clset{\tau_i}\}$, hence
\[\{ C \in \clset{\tau} \mid \formA \not \modelst C \}  
= \bigcup_{i=1}^n \{ l_i \vee  D \mid D \in \clset{\tau_i} \text{ and } \formA \not \modelst l_i \vee D \}.\]
By Proposition \ref{prop:model}(\ref{model:shift}), we deduce that
\[\{ C \in \clset{\tau} \mid \formA \not \modelst C \}  
= \bigcup_{i=1}^n \{ l_i \vee  D \mid  D \in \clset{\tau_i} \text{ and } \formA \wedge \compl{l_i} \not \modelst D \}.\] By the induction hypothesis, 
$\prune{\tau_i}{\formA \wedge \compl{l_i}}$ is an \trie and 
$\clset{\prune{\tau_i}{\formA \wedge \compl{l_i}}} = \{ D  \in \clset{\tau_i} \mid  \formA \wedge \compl{l_i} \not \models D \}$, thus:
\[\{ C \in \clset{\tau} \mid \formA \not \modelst C \}  
= \bigcup_{i=1}^n \{ l_i \vee D \mid D \in \clset{\prune{\tau_i}{\formA \wedge \compl{l_i}}} \}.\]
By definition, $\prune{\tau}{\formA} = 
\bigcup_{i=1}^n \{ \apair{l_i}{\prune{\tau_i}}{\formA \wedge \compl{l}} \}$.  The $l_1,\dots,l_n$ are pairwise distinct and all literals in $\clset{\tau_i}$ also occur in $\tau_i$, thus are greater than $l_i$. Hence this entails that 
$\prune{\tau}{\formA}$ is an \trie, and
\[\clset{\prune{\tau}{\formA}} = \bigcup_{i=1}^n \{ l_i \wedge D \mid D \in \clset{\prune{\tau_i}{\formA \wedge \compl{l}}} \} = \{ C \in \clset{\tau} \mid \formA \not \modelst C \}.\]
By definition, the literals in $\prune{\tau}{\formA}$ occur either in $l_1,\dots,l_n$ or $\prune{\tau_i}{\formA \wedge \compl{l}}$ (hence $\tau_i$), thus must occur in $\tau$. 
}
\end{itemize}
}
\begin{remark}
The {\trie}s may be represented as dags instead of trees.
In this case, it is clear that the complexity, defined as the number of satisfiability tests of forward subsumption (as defined in Lemma \ref{lem:forward}) is of the same order as the size of the dag, since the recursive calls only depend on the considered subtree.
For backward subsumption (see Definition \ref{def:backward}) the situation is different since the recursive calls have an additional parameter that is the formula $\phi$, which depends on the path in the \trie. The maximal number of satisfiability tests is therefore equal to the size of underlying tree, and not that of the dag. Note that it would be necessary to make copies of some of the subtrees, if two pruning operations are applied on the same (shared) subtree with different formulas. 
\end{remark}

\section{Experimental evaluation}

\label{sect:exp}
\input{impExp.tex}

\section{Conclusion}

\label{sect:conc}

We devised a generic algorithm to generate implicates modulo theories and showed that the corresponding implementation is
more efficient than a previous attempt based on superposition. 
This result was to be expected since the DPLL($\theory$) approach is more efficient than engines based on the Superposition Calculus for testing the satisfiability of quantifier-free formulas with a large combinatorial structure. Furthermore,
the used superposition engine had to be specifically tuned for implicate generation, and it is far less efficient than state-of-the-art 
systems such as Vampire \cite{RV01}, E \cite{Schulz:LPAR-2013} or Spass \cite{SPASS} (this is of course the advantage of having a generic algorithm using decision procedures as black boxes).
While our aim was to be completely generic, it is clear than the efficiency of the procedure could be improved 
in practice by integrating theory-specific algorithms for deriving consequences 
and normalizing formulas. For instance, in the case of the theory of equality with uninterpreted function symbols, the implicates could be normalized by replacing each term
by its minimal representative, as  is done in \cite{EPT17}.
Efficient, theory-dependent simplification procedures will also  be explored in future work.
A combination between the 
superposition-based approach \cite{EPT17}, in which the assertion of hypotheses is guided by the proof procedure could also be beneficial.
Our approach could also be combined with that of \cite{dillig2012minimum}, which is based on model building and quantifier-elimination.

\bibliography{biblio,Nicolas.Peltier}
\bibliographystyle{abbrv}
\end{document}

%% file: macros.tex
\newcommand{\BeginItemize}{\begin{itemize}}
\newcommand{\EndItemize}{\end{itemize}}

\newcommand{\isdef}{\stackrel{\mbox{\tiny def}}{=}}

\newcommand{\true}{\texttt{true}}
\newcommand{\false}{\texttt{false}}

\newcommand{\card}{\mathrm{card}}

\newcommand{\A}{{\cal A}}

\newcommand{\T}{{\cal T}}

\newcommand{\calP}{{\cal P}}



%% file: impExp.tex
Algorithm \ref{alg:impid} has been implemented in a C++ framework called \gpid.
The SMT solver
is used as a black box and \gpid can thus be plugged with any tool serving this purpose, provided an interface is written for it.
As a consequence, the handled theory is
only restricted by the SMT solver.
Three interfaces were implemented, respectively for \minisat \cite{DBLP:conf/sat/EenS03}, \cvcfsmt \cite{CVC4} and \ztsmt \cite{DBLP:conf/tacas/MouraB08}.
The implicate generator used in the reported experiments
is the one based on  \ztsmt, which turned out to be  more efficient  on the considered benchmarks.
All the tests were run on one core of an Intel(R) Core(TM) i5-4250U machine running at 1,9 GHz with 1 GiB of RAM.
The benchmarks are extracted from the SMTLib \cite{BarST-SMTLIB} library, 
the considered theories are quantifier-free uninterpreted functions (QF\_UF) and quantifier-free linear integer arithmetic with uninterpreted functions (QF\_UFLIA).
For obvious reasons, only satisfiable examples have been kept for analysis.
Abducible literals are
part of the problem input, they
are generated by considering all ground equalities and disequalities with a maximal depth provided by the user;
all the experiments were conducted using a maximal depth of $1$ and the average number of abducible literals is around $13397$ (min. $1741$, max. 17.$10^6$).
We chose not to apply unit propagation simplifications to the considered sets of clauses.
More precisely,
this means that we let $U=M$ at line \ref{line:propsimp} of Algorithm \ref{alg:impid} and delegate the
simplifications that could occur in the following line to the satisfiability checker.
The reason for this decision is that efficiently performing such simplifications can be difficult and strongly depends on
the theory.
We also define $\fixlit{}{}{}$ as the complementation on literals and $\calP$ as either $\true$ or a predicate ensuring $\card{M} \leq n$ to generate {\timplicate}s of size at most  $n$.
In all the experiments, the prime implicates filter ($\submin$) was not active, so that implicates can be generated on the fly.
Finally, if available, we recover models of $S \cup M$ from the SMT solver in order to further prune the set of abducibles (see Line \ref{line:modelsimp} of Algorithm \ref{alg:impid}).
\begin{table}[t]
    \caption{Number of problems for which at least one \timplicate of a given maximal size  can be generated in a given amount of time (in seconds), for the QF\_UF SMTLib benchmark (2549 examples). \label{fig:gpid:uf:first}}
    	\begin{center}
    		\setlength{\tabcolsep}{4pt}
    	\begin{tabular}{|c|c|c|c|c|c|c|c|c|}
            \hline
            \backslashbox{Size}{Time}&$[0,0.5[$&$[0.5,1[$&$[1,1.5[$&$[1.5,2[$&$[2,5[$&$[5,10[$&$[10,35[$&None\\\hline
            $1$&$2235$&$75$&$28$&$16$&$33$&$32$&$61$&$69$\\\hline
            $2$&$2236$&$81$&$27$&$16$&$30$&$23$&$67$&$69$\\\hline
            $3$&$2236$&$79$&$27$&$16$&$34$&$23$&$65$&$69$\\\hline
            $4$&$2230$&$84$&$23$&$18$&$33$&$24$&$68$&$69$\\\hline
            $5$&$2231$&$79$&$27$&$12$&$36$&$22$&$73$&$69$\\\hline
            $6$&$2234$&$73$&$29$&$15$&$30$&$24$&$75$&$69$\\\hline
            $7$&$2231$&$81$&$23$&$15$&$33$&$22$&$75$&$69$\\\hline
            $8$&$2233$&$78$&$23$&$16$&$33$&$21$&$76$&$69$\\\hline
            
        \end{tabular}
    \end{center}
\end{table}
\begin{table}[t]
    \caption{Number of problems for which at least one \timplicate of a given maximal size  can be generated in a given amount of time (in seconds),  for the QF\_UFLIA SMTLib benchmark (400 examples). \label{fig:gpid:uflia:first}}
    \begin{center}
    	\setlength{\tabcolsep}{4pt}
        \begin{tabular}{|c|c|c|c|c|c|c|c|c|}
            \hline
            \backslashbox{Size}{Time}&$[0,0.5[$&$[0.5,1[$&$[1,1.5[$&$[1.5,2[$&$[2,5[$&$[5,10[$&$[10,35[$&None\\\hline
            $1$&$120$&$23$&$46$&$76$&$100$&$6$&$25$&$4$\\\hline
            $2$&$120$&$23$&$6$&$0$&$0$&$0$&$247$&$4$\\\hline
            $3$&$120$&$23$&$6$&$0$&$96$&$4$&$147$&$4$\\\hline
            $4$&$120$&$23$&$6$&$0$&$0$&$0$&$247$&$4$\\\hline
            $5$&$120$&$23$&$6$&$0$&$0$&$0$&$247$&$4$\\\hline
            $6$&$120$&$22$&$7$&$0$&$0$&$0$&$247$&$4$\\\hline
            $7$&$121$&$22$&$6$&$0$&$0$&$0$&$247$&$4$\\\hline
            $8$&$116$&$24$&$6$&$3$&$0$&$0$&$247$&$4$\\\hline
            
        \end{tabular}
    \end{center}
\end{table}
 Tables 1 and 2 show the number of examples for which our tool generates at least one \timplicate for a given timespan, for the QF\_UF and QF\_UFLIA benchmarks respectively. The results show that our tool is quite efficient, since it fails to generate any \timplicate within $35$ seconds for only 2\% (resp. 1\%) of the QF\_UF (resp. QF\_UFLIA) benchmarks.
Figure
\ref{fig:gpid:uflia:exists} shows the proportion of the
QF\_UFLIA set for which \gpid generates an implicate in less than $15$ seconds, depending on the maximal size constraint.
For the QF\_UF benchmark, the proportion decreases from $97\%$ for a maximal size constraint of $1$ to $95\%$ when there are no size restrictions.
We also point out that for $57\%$ of the QF\_UF benchmark, we are actually able to generate all the {\timplicate}s of size $1$ in less than $15$ seconds.
\begin{wrapfigure}{r}{0.5\textwidth}
    \caption{Proportion (out of 100) of examples of the QF\_UFLIA benchmark where \gpid generates at least one implicate under $15$ seconds. \label{fig:gpid:uflia:exists}}
    \begin{center}
        \includegraphics[scale=0.35]{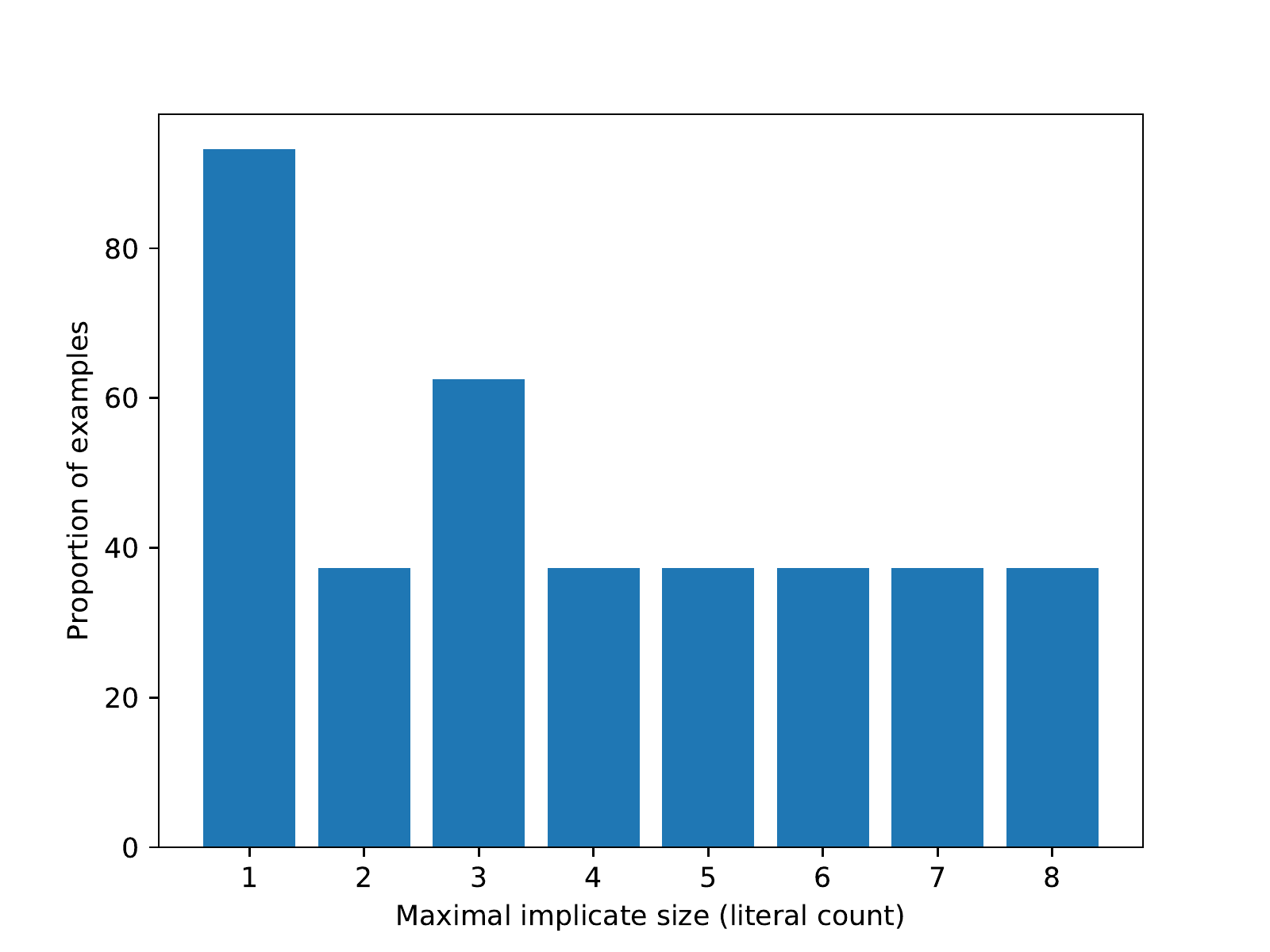}
    \end{center}
\end{wrapfigure}

We ran additional experiments to compare this approach with a previous one based on a superposition-based approach \cite{EPT15,EPT17} and implemented in the {\csp} tool.
As far as we are aware, \csp is the only other available tool for implicate generation in the theory of equality with uninterpreted function symbols.
Previous experiments (see, e.g., \cite{EPT15,EPT17}) showed that \csp is  already more efficient than 
approaches based on a reduction to propositional logic for generating implicates of ground  equational formulas, which is why we did not run comparisons against tools for propositional implicate generation.
\csp is based on a constrained calculus defined by the usual inference rules of the superposition calculus
together with additional rules to dynamically assert new abducible hypotheses on demand into the search space. 
The asserted hypotheses are attached to the clauses as constraints and, when an empty clause is generated, 
the negation of these hypotheses yields a \timplicate.
We chose to compare the tools by focusing on their ability to generate at least one  \timplicate of a given size.
Indeed, generating
all (prime) {\timplicate}s is unfeasible within a reasonable amount of time except for very simple formulas, and comparing the raw number of {\timplicate}s generated
is not relevant because some of these may actually be redundant w.r.t.\ non-generated ones\footnote{A refined comparison 
of the set of generated  {\timplicate}s modulo theory entailment is left for future work.}.
 We believe in practice, being able to efficiently compute a small number of {\timplicate}s for a complex problem is more useful than 
computing huge sets of {\timplicate}s but only for simple formulas.
The following experiments are only based on benchmarks that can be solved by both prototypes, as \csp is not capable of handling integer arithmetics.
\begin{figure}[t]
    \caption{Number of examples from the QF\_UF benchmark set for which \gpid(on the left, darker color) and \csp (on the right, lighter color) generate at least one \timplicate within a given time (a)  and generate at least one implicate of a given maximal size under $15$ seconds (b)\label{fig:compare:uf:al1}}
    \vspace{-1.5em}
    \begin{center}
        \begin{tabular}{cc}
            (a) & (b) \\
            \vspace{-1em}
            \includegraphics[scale=0.35]{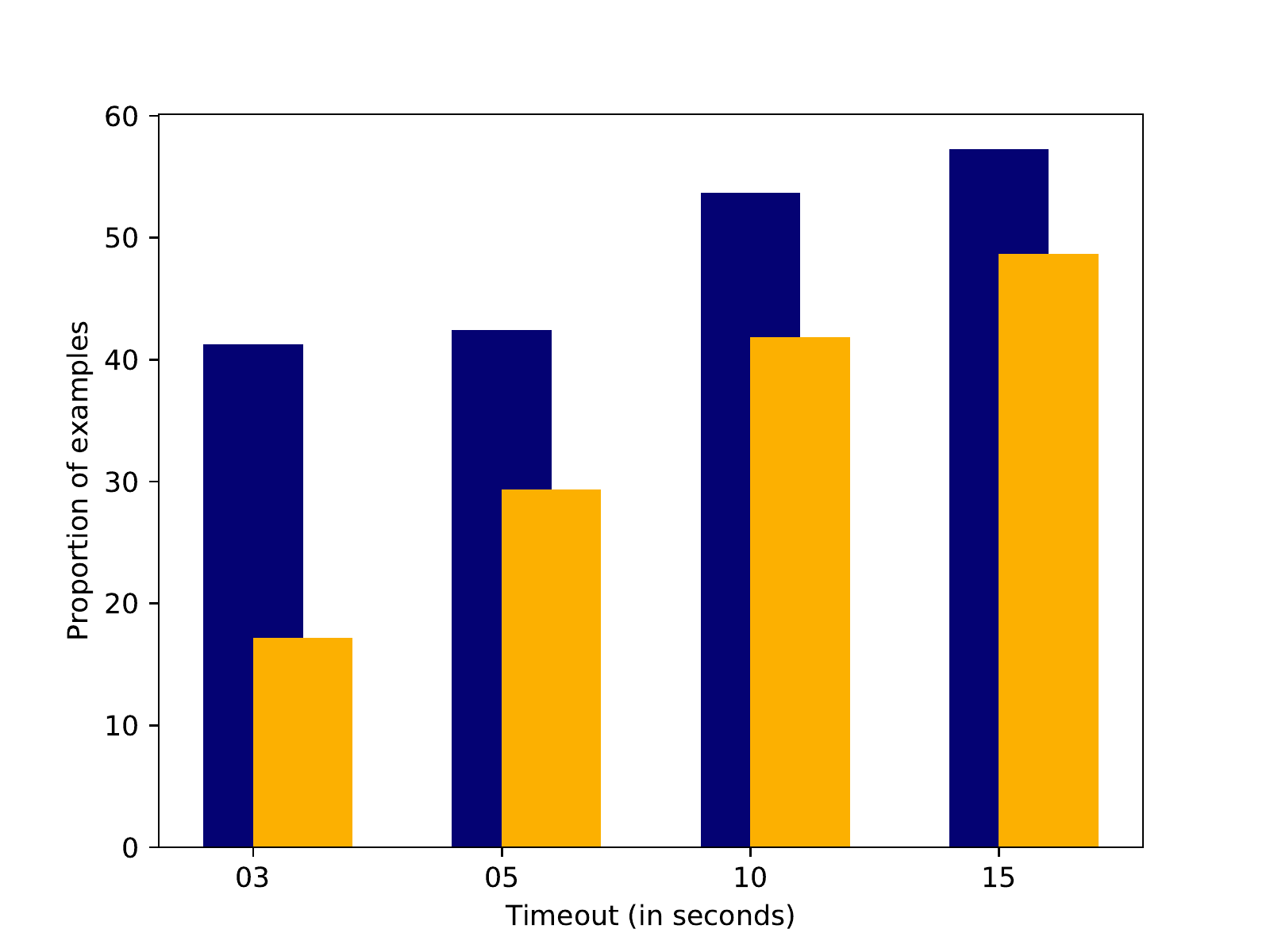}
            &
            \vspace{-1em}
            \includegraphics[scale=0.35]{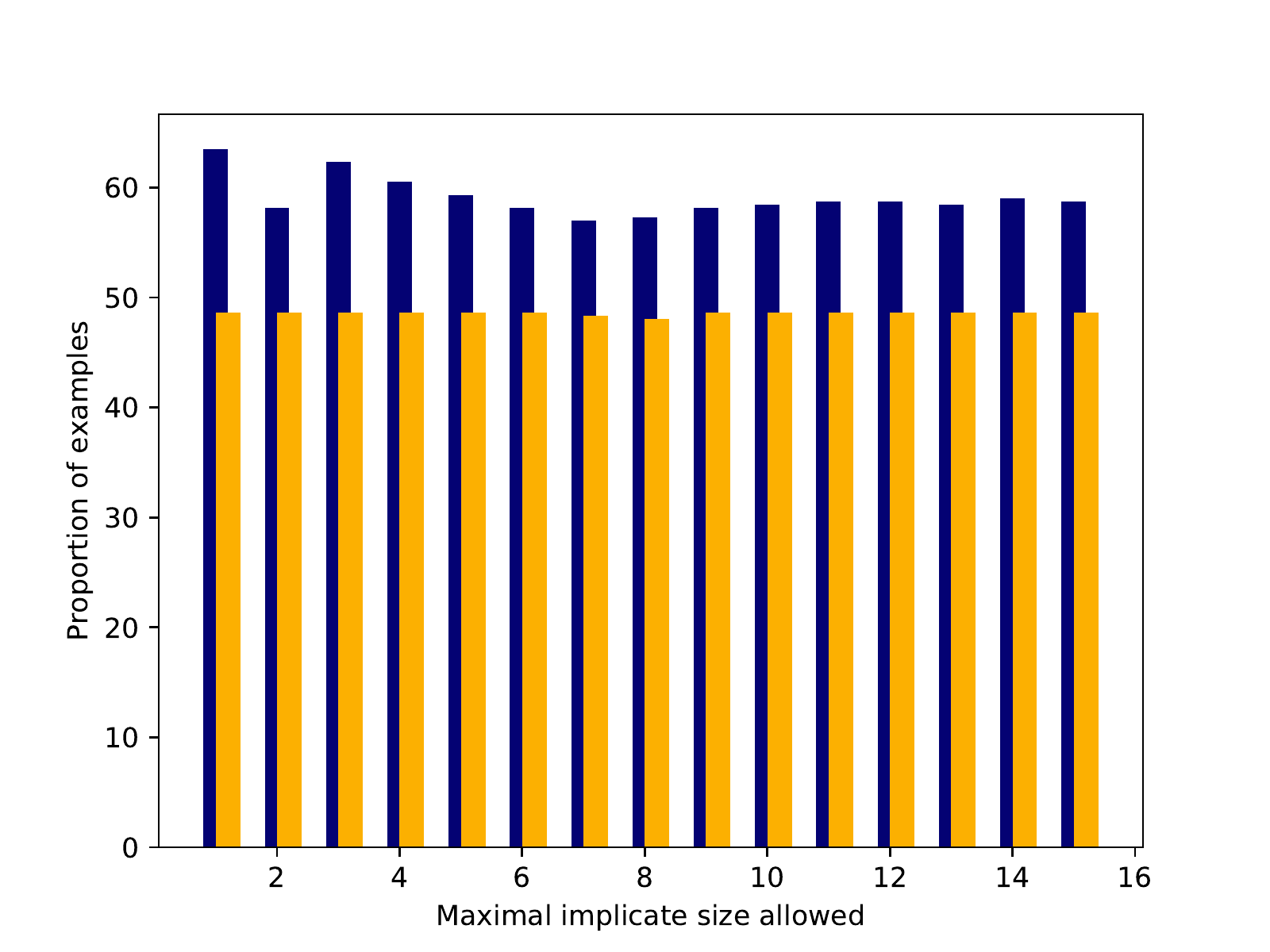}
        \end{tabular}
    \end{center}
\end{figure}
We represented on Figure \ref{fig:compare:uf:al1} the number of examples for which both tools can generate at least one \timplicate with a given maximal size constraint for various timeouts (a) and generate at least one \timplicate within a given time limit for various maximal size constraints (b).